\documentclass[11pt]{article}
\usepackage{epsf}
\usepackage{latexsym}
\usepackage{amssymb,amsmath,amsfonts}

\def \vs {\vspace}
\def \ni {\noindent}

\def \md {\models}
\def \Sg {\Sigma}
\def \sg {\sigma}
\def \wed {\wedge}
\def \se {\subseteq}
\def \ep {\epsilon}
\def \hs {\hspace}

\newtheorem{theorem}{Theorem}[section]

\newtheorem{lemma}{Lemma}[section]

\newtheorem{example}{Example}[section]

\newenvironment{proof}{\noindent{\bf Proof.}}{\bigskip}

\newlength{\figwidth}
\begin{document}

\title{An Incremental Knowledge Compilation in First Order Logic}

\author{Manoj K. Raut\\ School of
 Computing Sciences\\ VIT University, Vellore
 \\ Email: {\tt manoj.raut@gmail.com}}

\date{}
\maketitle

\begin{abstract}
An algorithm to compute the set of prime implicates of a
quantifier-free clausal formula $X$ in first order logic had been
presented in earlier work. As the knowledge base $X$ is dynamic, new
clauses are added to the old knowledge base. In this paper an
incremental algorithm is presented to compute the prime implicates of $X$
and a clause $C$ from $\pi(X)\cup C$. The
correctness of the algorithm is also proved.
\end{abstract}


\section{Introduction}

 Propositional entailment is a central issue in artificial
intelligence due to its high complexity. Determining the logical
entailment of a given query from a knowledge base is intractable in
general \cite{Cook} as all known algorithms run in time exponential in
the size of the given knowledge base. To overcome such computational
intractability, the propositional entailment problem is split into two
phases such as {\em off-line} and {\em on-line}. In the off-line phase
the original knowledge base $X$ is transformed into another knowledge
base $X^{'}$ and the queries are answered in the on-line phase from
the new knowledge base in polynomial time in the size of $X^{'}$. In
such type of compilation most of the computational overhead shifted
into the off-line phase is amortized over on-line query
answering. The off-line computation is known as {\em knowledge
compilation}.

Several algorithms in knowledge compilation have been suggested so
far, see for example, \cite{Coudert1, Jackson1, Kleer1, Kleer2, Kean1,
Ngair, Reiter, Shiny, Slagle, Strzemecki, Tison, delVal1}. In these
approaches of knowledge compilation, a knowledge base $KB$ is compiled
off-line into another {\it equivalent} knowledge base $\pi(KB)$, i.e,
the set of prime implicates of $KB$, so that queries can be answered
from $\pi(KB)$ in polynomial time. Most of the work in knowledge
compilation have been restricted to propositional knowledge bases in
spite of the greater expressing capacity of a first order knowledge
base. Due to lack of expressive power in propositional logic, first
order logic is required to represent knowledge in many problems.  An
algorithm to compute the set of prime implicates of a first order
formula in SCNF had been proposed in \cite{Manoj}.

As a knowledge base is not static, new clauses are added to the
existing knowledge base. It will be inefficient to compute the set of
prime implicates of the new knowledge base from the scratch. On the
other hand properties of the old $\pi(KB)$ can be utilized for
computing the new $\pi(KB)$. In this paper, we suggest an incremental
method to compute the set of prime implicates of the new knowledge
base from the prime implicates of the old knowledge base. The
incremental method based upon the algorithm discussed in \cite{Manoj}.

  The paper is organized as follows. We present the definitions in
  Section 2. In Section 3, we describe the properties and main results
  for computing the prime implicates incrementally. In Section 4, we
  present the incremental algorithm and its correctness. Section 5
  concludes the paper.

\section{Preliminary Concepts}

\ni The alphabet of first order language contains the symbols
$x,y,z,\ldots$ as variables, $f,g,h,\ldots$ as function symbols,
$P,Q,R,\ldots$ as predicates, $\neg,\vee,\wed$ as connectives, $(,)$
and `,' as punctuation marks and, $\forall$ as the universal
quantifier. We assume the syntax and semantics of first order
logic. For an interpretation or a first order structure $i$ and a
formula $X$, we write $i\md X$ if $i$ is a model of $X$. For a set of
formulas $\Sg$ (or a formula) and any formula $Y$ we write $\Sg\md Y$
to denote that for every interpretation $i$ if $i$ is a model of every
formula in $\Sg$ then $i$ is a model of $Y$. In this case, we call $Y$
a logical consequence of $\Sg$. When $\Sg=\{X\}$, we write $X\md Y$
instead of $\{X\}\md Y$. If $X\md Y$ and $Y\md X$ then $X$ and $Y$
are said to be equivalent which is denoted by $X\equiv Y$.

A literal is an atomic formula or negation of an atomic formula. A
disjunctive clause is a finite disjunction of literals which is also
represented as a set of literals. A quantifier-free formula is in
conjunctive normal form (CNF, infact, SCNF) if it is a finite conjunction of
disjunctive clauses. For convenience, a formula is also represented as
a set of clauses.  In this paper, we consider formulas only in clausal
form. In this representation, all variables are considered universally
quantified.

Two literals $r$ and $s$ are said to be {\it complementary} to
each other iff the set $\{r,\neg s\}$ is unifiable with respect to
a most general unifier $\xi$. We call $\xi$ a complementary
substitution of the set $\{r,\neg s\}$. For example, $Pxf(a)$ and
$\neg Pby$ are complementary to each other with respect to the
complementary substitution (most general unifier or mgu, for
short) $[x/b,y/f(a)]$. So a most general unifier bundles upon
infinite number of substitutions to a finite number.

A clause which does not contain a literal and its negation is said
to be {\it fundamental}. Thus a non-fundamental clause is valid.
We avoid taking non-fundamental clauses in clausal form because the
universal quantifiers appearing in the beginning of the formula
can appear before each conjunct of the CNF since $\forall$
distributes over $\wed$. So each clause in a formula of the
knowledge base is assumed to be non-fundamental.

 Let $C_{1}$ and $C_{2}$ be two disjunctive clauses. Then $C_{1}$ {\em
 subsumes} $C_{2}$ iff there is a substitution $\sigma$ such that
 ${C_{1}}{\sigma}\subseteq C_{2}$, i.e, ${C_{1}}{\sigma}\models
 C_2$. For example, $\{\neg{R}{x}{f}(a), \neg{P}{y}\}$ subsumes
 $\{\neg{R}{g}(a){f}(a), \neg{P}{y},$ ${Q}{z}\}$ with respect to a
 $\sigma=[x/{g}(a)]$. A disjunctive clause $C$ is an {\it implicate}
 of a finite set of formulas $X$ (assumed to be in CNF) if $X\sg\md C$
 for a substitution $\sg$. We write $I(X)$ as the set of all
 implicates of $X$. A clause $C$ is a {\it prime implicate} of $X$ if
 $C$ is an implicate of $X$ and there is no other implicate $C^{'}$ of
 $X$ such that $C^{'}\tau\md C$ for a substitution $\tau$ (i.e., if no
 other implicate $C^{'}$ subsumes $C$). $\Pi(X)$ denotes the set of all
 prime implicates of $X$. It may be observed that if $C$ is not prime
 then there exists a prime implicate $D$ of $X$ such that $D\tau\md
 C$. The set of all implicates of $X$ is denoted by $\Psi(X)$

 Note that the notion of prime implicate is well
 defined as the knowledge base contains clauses unique up to
 subsumption. Let $Y$ be a set of fundamental clauses. The residue of
 subsumption of $Y$, denoted by ${Res}(Y)$ is a subset of $Y$ such
 that for every clause $C\in Y$, there is a clause $D\in {Res}(Y)$
 where $D$ subsumes $C$; and no clause in ${Res}(Y)$ subsumes any
 other clause in ${Res}(Y)$.

Let $C_1$, $C_2$ be two clauses in $X$ and $r\in C_1$, $s\in C_2$ be a
pair of complementary literals with respect to a most general unifier
$\sg$. The resolution of $C_1$ and $C_2$ is
$C=[(C_1-\{r\})\cup(C_2-\{s\})]\sg$. If $C$ is fundamental, it is
called {\em consensus} of $C_1$ and $C_2$ denoted by $CON(C_1,C_2)$.
$C$ can also be written as $[(C_1\sg-\{t\})\cup(C_2\sg-\{\neg t\})]$
for a literal $t$, provided $r\sg=t$ and $s\sg=\neg t$. We can also
say that $C$ is the propositional consensus of $C_1\sg$ and $C_2\sg$.
For example, if $C_{1}=\{{R}{b}{x}, \neg {Q}{g}(a)\}$ and
$C_{2}=\{{R}{a}{b}, {Q}{z}\}$ then $CON(C_{1}, C_{2})=\{{R}{b}{x},
{R}{a}{b}\}=$ propositional consensus of $C_1[z/g(a)]$ and
$C_2[z/g(a)]$. If $C$ is the consensus of $C_1$ and $C_2$ with respect
to a most general unifier $\sg$ then $C$ is said to be associated with
$\sg$. By default, each clause $C$ of a set of formulas $X$ is
associated with the empty substitution $\ep$. Let $C_1$ and $C_2$ be
two resolvent clauses associated with substitutions $\sg_1$ and
$\sg_2$, respectively.  Then their consensus with respect to $\sg$ is
defined provided $\sg_1\sg=\sg_2\sg$. In that case the consensus is
the propositional consensus of $C_1\sigma$ and $C_2\sigma$ and the
consensus is associated with the substitution
$\sigma_{3}=\sigma_{1}{\sigma}=\sigma_{2}{\sigma}$.

\section{Computation of Prime Implicates}

\ni Besides presenting some main results from \cite{Manoj}, we
 describe the computation of prime implicates incrementally of
 quantifier free first order formulas in clausal form. Let $X=\{C_{1},
 \ldots, C_{n}\}$ be a formula where each disjunctive clause $C_{i}$
 is fundamental. Each $C_{i}$ is an implicate of $X$ with respect to
 the empty substitution, but each may not be a prime implicate. The
 key is the subsumption of implicates of $X$. As the clauses we deal
 with are disjunctive, if $C_{1}$ subsumes $C_{2}$ then there is a
 substitution $\sigma$ such that ${C_{1}}{\sigma}\models C_2$. We will
 see that computation of prime implicates is the result of deletion of
 subsumed clauses from the consensus closure. We also explore the
 relationship between consensus closure and prime implicates of a
 formula. We can derive the following two results from \cite{Manoj}.

\begin{lemma}
\label{prime11}  A clause $D$ is an implicate of $X\cup C$ if and
only if there is a prime implicate $D^{'}$ of $X\cup C$ such that
$D^{'}$ subsumes $D$.
\end{lemma}

\begin{lemma}
\label{prime12} $X\cup C\equiv \Psi(X\cup C)\equiv\Pi(X\cup C)$
\end{lemma}

The computational aspects of prime implicates is given below. For a
set of clauses $X$, let ${L}(X)$ be the set of all consensus among
clauses in $X$ along with the clauses of $X$, {\it i.e.},
${L}(X)=X\cup \{ S : S$ is a consensus of each possible pair of
clauses in $X \}.$ We construct the sequence $X, {L}(X), {L}({L}(X)),
\ldots$, i.e, ${L^{0}}(X)=X$, ${L^{n+1}}(X)={L}({L^{n}}(X))$ for
$n\geq 0$.  Define the {\it consensus closure} of $X$ as
$\overline{L}(X)= \cup\{{L^{i}}(X):i\in \mathbb N\}$.

\begin{example}
Let $X=(Pxa\vee\neg Qaf(x))\wed(\neg Pba\vee Rbz)\wed(\neg Rxf(a)\vee
Qzf(a))=C_1\wed C_2\wed C_3$. 
\end{example}

The consensus of $C_1$ and $C_2$ with respect to the substititon
$[x/b]$ is $C_4=(\neg Qaf(b)\vee Rbz)$, of $C_1$ and $C_3$ with
respect to the substitution $[z/a, x/a]$ is $C_5=(Paa\vee \neg
Raf(a))$, of $C_2$ and $C_3$ with respect to the substitution
$[x/b,z/f(a)]$ is $C_6=(\neg Pba\vee Qf(a)f(a))$. So $C_4$, $C_5$,
$C_6$ are associated with substitutions $[x/b]$, $[z/a,x/a]$, $[x/b,
z/f(a)]$ and each of $C_1$, $C_2$, $C_3$ is associated with the empty
substitution $\epsilon$. Then \\

\noindent $L^{1}(X)= (Pxa\vee\neg Qaf(x))\wed(\neg Pba\vee
Rbz)\wed(\neg Rxf(a)\vee Qzf(a))\wed $ \\ 
$~~~~~~~~~~~~~$ $(\neg
Qaf(b)\vee Rbz)\wed (Paa\vee \neg Raf(a))\wed (\neg Pba\vee
Qf(a)f(a))$\\
$~~~~~~~~~~~~$ $=C_1\wed C_2\wed C_3\wed C_4\wed C_5\wed C_6$.\\

The consensus of $C_3$ and $C_4$ with respect to the
substitution $[x/b, z/f(a)]$ is $C_7=(Qf(a)f(a)\vee \neg Qaf(b))$. So
$C_7$ is associated with the substitution $[x/b,z/f(a)]$. Note that
the consensus of $C_1$ and $C_6$ with respect to the substitution
$[x/b]$ is not possible as the composition of substitution is not well
defined, i.e, $[\epsilon][x/b]\neq [x/b,z/f(a)][x/b]$. Similarly the
consensus between $C_3$ and $C_5$ is not possible as the composition
of substitution is not well defined. Hence,\\

\noindent $L^{2}(X)= ((Pxa\vee\neg
Qaf(x))\wed(\neg Pba\vee Rbz)\wed(\neg Rxf(a)\vee Qzf(a))\wed (\neg
Qaf(b)\vee $\\
$~~~~~$ $ Rbz)\wed (Paa\vee \neg Raf(a))\wed (\neg Pba\vee
Qf(a)f(a))\wed (Qf(a)f(a)\vee \neg Qaf(b)))$.\\

\noindent Since
$L^{2}(X)=L^{3}(X)$, we have, $\overline{L}(X)= L^{2}(X)$.\hfill{$\Box$}

As each clause of a formula $X$ is itself an implicate, the following
result shows that other implicates can be computed by taking
consensus among the clauses of a formula $X$.

\begin{theorem}
\label{prime13} Consensus of two implicates of $\pi(X_1)\cup X_2$ 
is an implicate of the formula $X_1\cup X_2$ where $X_1$ and $X_2$ are
sets of clauses.
\end{theorem}

\begin{proof}
let $C_1$ and $C_2$ be two implicates of $\pi(X_1)\cup X_2$ associated
with $\sg_1$ and $\sg_2$ respectively.  $(\pi(X_1)\cup X_2)\sg_1\md
C_1$, $(\pi(X_1)\cup X_2)\sg_2\md C_2$. $CON(C_{1},
C_{2})=PCON(C_{1}\sigma, C_{2}\sigma)$ provided
${\sigma_{1}}{\sigma}={\sigma_{2}}{\sigma}$ for some substitution
$\sigma$.  So $C=CON(C_{1},C_{2}) =
((C_{1}\sigma-\{t\})\,\cup\,(C_{2}\sigma-\{\neg t\}))$. Let
$i\md(X_1\cup X_2)$.  $i\md (X_1\cup X_2)\sg = X_1\sg\cup X_2\sg$. So
$i\md X_1\sg$ or $i\md X_2\sg$. Let $i\md X_1\sg$. Then,
$i\md\pi(X_1)\sg$ (by Lemma \ref{prime12}. Moreover, $i\md \pi(X_1)\sg\cup
X_2\sg$. $i\md (\pi(X_1)\cup X_2)\sg$.  $i\md C_1$ and $i\md
C_1\sg$. Similarly, if $i\md X_2\sg$ then $i\md C_2\sg$. Suppose $i\md
t$. Then $i\md C_2\sg-\{\neg t\}$, i.e., $i\md
C$. Similarly, if $i\md\neg t$, then $i\md
C_1\sg-\{t\}$, i.e, $i\md c$. Hence $C$ is an implicate of $X_1\cup X_2$
\hfill{$\Box$}
\end{proof}

The following result tells that we can add the prime implicates one by
one to $X$ as it preserves the truth value.

\begin{theorem}
\label{prime13_1}
Let $X=\{C_1,C_2,\ldots,C_n\}$ be a formula and $C$ be the consensus
of a pair of clauses from $X$ then $X\equiv X\wed C$.
\end{theorem}
\begin{proof}
As $X$ is in SCNF and $C$ is a disjunctive clause, $X\wed C\md
X$. Conversely, let $i\md X$. Then $i\md C_k$ for $1\leq k\leq n$, and $i\md
C_i\wed C_j$ for $1\leq \{i, j\}\leq n$. Let $C= CON(C_i, C_j)$, where
$r$ and $s$ be a pair of complementary literals and $r\in C_i$,
$s\in C_j$, $r\sg=t$ and $s\sg=\neg t$.  Let $i\md C_i\wed
C_j=(D_1\vee r)\wed(D_2\vee s)$, where $D_1$ is a disjunctive clause
in $C_i$ leaving $r$ and $D_2$ is a disjunctive clause from $C_j$
leaving $s$. $C= CON(C_i, C_j)=CON(D_1\vee r, D_2\vee s)=D_1\sg\vee
D_2\sg$. Further, $i\md (D_1\vee r)$ and $i\md(D_2\vee s)$. If $i\md D_1\vee r$
then $i\md D_1$. $i\md D_1\sg$. $i\md D_1\sg\vee D_2\sg$ and $i\md
C$. Similarly if $i\md D_2\vee s$ then $i\md C$. Let $i\md r$. Then
$i\not\md s$, so $i\md D_2$ and $i\md D_2\sg$. This implies, $i\md C$. If
$i\md s$ then $i\not\md r$, $i\md D_1$. $i\md D_1\sg$, and $i\md
C$. This implies $i\md X\wed C$. \hfill{$\Box$}

\end{proof}

\begin{theorem}
\label{prime13_2}
$Res(\overline{L}(\pi(X)\cup C))= Res(\overline{L}(X\cup C))$
\end{theorem}
\begin{proof}
Obviously, $\overline{L}(\pi(X)\cup C)\se \overline{L}(X\cup C)$. This
implies $Res(\overline{L}(\pi(X)\cup C))\se Res(\overline{L}(X\cup C))$

Conversely, let $C_1\in Res(\overline{L}(X\cup C))$. So $C_1\in
\overline{L}(X\cup C)$ and there exists no $D\in \overline{L}(X\cup
C)$ such that $D$ subsumes $C_1$. If $C_1\not\in
Res(\overline{L}(\pi(X)\cup C))$ then $D_1\in \overline{L}(\pi(X)\cup
C)$ such that $D_1$ subsumes $C_1$. As $\overline{L}(\pi(X)\cup C)\se
\overline{L}(X\cup C)$, there exists $D_1\in \overline{L}(X\cup C)$
such that $D_1$ subsumes $C_1$. This gives a contradiction. So $C_1\in
Res(\overline{L}(\pi(X)\cup C))$. This implies $\overline{L}(X\cup
C)\se Res(\overline{L}(\pi(X)\cup C))$.\hfill{$\Box$}
\end{proof}

The following results are consequences of Lemma \ref{prime11} and
\ref{prime12} and Theorem \ref{prime13_2}.

 \begin{theorem}
\label{prime14} A clause $D_1$ is an implicate of $X\cup C$ iff 
there is $D_2\in\overline {L}(\pi(X)\cup C)$ such that $D_2$ subsumes
$D_1$.
\end{theorem}
\begin {theorem}
\label{prime15} The set of all prime implicates of $X\cup C$ is a subset of 
the consensus closure of
$\pi(X)\cup C$, i.e, $\pi(X\cup C)\se \overline
{L}(\pi(X)\cup C)$. Moreover, $\pi(X\cup C)= Res(\overline
{L}(\pi(X)\cup C))$
\end{theorem}

Moreover, the sets $\pi(X_1\cup X_2)$ and $\pi(\pi(X_1)\cup X_2)$ are
not only equivalent but also identical, as the next result shows.
theorem.

\begin{theorem}
\label{prime16} $\pi(X_1\cup X_2)=\pi(\pi(X_1)\cup X_2)$
\end{theorem}
\begin{proof}
let $C\in \pi(X_1\cup X_2)$. $(X_1\cup X_2)\sg_1\md C$ and there does
not exist any implicate $D$ of $X_1\cup X_2$ such that $D$ subsumes
$C$. $X_1\sg\cup X_2\sg\md C$ and there doesnot exist any implicate
$D$ of $X_1\cup X_2$ such that $D$ subsumes $C$. As $X_1\equiv
\pi(X_1)$,  $\pi(X_1)\sg\cup X_2\sg\md C$ and there does not exist
any implicate $D$ of $\pi(X_1)\cup X_2$ such that 
$D$ subsumes $C$. $(\pi(X_1)\cup X_2)\sg\md C$ and there does not exist
any implicate $D$ of $\pi(X_1)\cup X_2$ such that 
$D$ subsumes $C$. $C\in\pi(\pi(X_1)\cup X_2)$. So $\pi(X_1\cup X_2)\se
\pi(\pi(X_1)\cup X_2)$. Similarly the inclusion $\pi(\pi(X_1)\cup
X_2)\se \pi(X_1\cup X_2)$ is proved. \hfill{$\Box$}
\end{proof}

\section{Incremental Algorithm}
The following algorithm computes the set of prime implicates of
$\pi(X_1)\wed \Sg$ (i.e, of $(\pi(X_1)\cup \Sg)$ by consensus
subsumption method in first order logic.  Recall that for a set of
clauses $A$, $L(A)$ denotes the set of clauses of $A$ along with the
consensus of each possible pair of clauses of $A$. The algorithm
computes the consensus $L(\pi(X_1)\cup \Sg)$ by taking clauses from
$\pi(X_1)$ and clauses from $\Sg$. It does not compute the consensus
between two clauses of $\pi(X_1)$ as it is wasteful. $L(\pi(X_1)\wed
\Sg)=\{CON(D_1, D_2)$ : $D_1\in \pi(X_1)\cup \Sg$ and $D_2\in \Sg\}$. The
algorithm applies subsumption on $L(\pi(X_1)\wed \Sg)$ and keeps the
residue $Res(L(\pi(X_1)\wed \Sg))$ and repeat the steps till two
iteration steps produce the same result.


\newpage

\ni{\bf Algorithm} \hs{.2cm} {\em INCRPI}

\vs{.2cm}

\ni  Input: The set of prime implicates $\pi(X_1)$ and a clause $C$\\
 Output: The set of prime implicates of $X_1\cup C$\\
{\tt begin\\
$~~~~$if $C$ is a non-fundamental clause, then\\
$~~~~~~~~$$\pi(X_1\cup C)=\pi(X_1)$\\
$~~~~$else\\
$~~~~~~~~$$\Sg=\{C\}$;\\
$~~~~~~~~$$\eta_0=\emptyset$;\\
$~~~~~~~~$$i=1$;\\
$~~~~~~~~$$\gamma=\pi(X_1)\cup \Sg$;\\
$~~~~~~~~$$\eta_i=Res(\pi(X_1)\cup\Sg)$;\\
$~~~~~~~~$if $\Sg$ is deleted\\
$~~~~~~~~~~$then stop\\
$~~~~~~~~$else\\
$~~~~~~~~~~$while $\eta_i\neq \eta_{i-1}$\\
$~~~~~~~~~~~~~$compute $R=CON(D_1,D_2)$ s.t. 
$D_1\in \eta_i$ and $D_2\in \Sg$;\\
$~~~~~~~~~~~~~$$\Sg=\Sg\cup R$;\\
$~~~~~~~~~~~~~$$L(\eta_i)=\eta_i\cup \Sg$;\\ 
$~~~~~~~~~~~~~$$\eta_{i+1}=Res(L(\eta_i))$;\\
$~~~~~~~~~~~~~$if any clause of $\Sg$ is deleted\\
$~~~~~~~~~~~~~~~~~~~~$then update $\Sg$;\\
$~~~~~~~~~~~~~$$i=i+1$;\\
$~~~~~~~~~~$$\pi(\gamma)=\eta_{i+1}$;\\
end}

\begin{theorem}
\label{correctness} Let $\pi(X_1)$ be a set of prime implicates of a 
formula $X_1$ and $C$
be any clause. The incremental algorithm generates the set of prime
implicates of $X_1\cup \{C\}.$
\end{theorem}
\begin{proof}
Let $\gamma=\pi(X_1)\cup \{C\}$. $\eta_1$ is computed
by taking residue of subsumption on $\gamma$. If a clause
$C\in \gamma$ subsumes a clause $D \in \gamma$,
any clause entailed by $D$ is entailed by $C$. So $D$ can be
discarded from $\gamma$ without changing its deduction closure.
Let $\eta_1= Res(\gamma)$. 
$\gamma =\pi(\eta_1)$, by Theorem \ref{prime15}. 

If the clause $C$ is deleted while taking residue then
$\pi(\gamma)=\pi(X_1\cup \{C\})=\pi(X_1)$.  If $C$ is not deleted from
$\Sg$, then the algorithm computes the consensus $R$ between a pair of
clauses from $\eta_1$ (i.e,$Res(\gamma)$) and $C$. It does not take
consensus between two clauses in $\pi(X_1)$ as they are prime
implicates. Any attempt to take consensus between them will increase
the number of uncessary operations. As every clause of $\eta_1$ is an
implicate of $\eta_1$, by Theorem \ref{prime13}, $R$ is an implicate
of $\eta_1$, i.e, of $\gamma$. $R$ can be added to $\eta_1$ without
changing its deduction closure by Theorem \ref{prime13_1}. We add $R$
to $\Sg=\{C\}$ as the new clauses formed can also take part in further
consensus. We maintain $\pi(X_1)$ and $\Sg$ separately so that while
taking consensus next time at least one clause will be from $\Sg$,
i.e, $L(\eta_1)=\eta_1\cup \Sg$. By Theorem \ref{prime15},
$\pi(\gamma)=\pi(\eta_1)\se\overline{L}(\eta_1)$. If any clause $C$
subsumes a clause $D$ in $L(\eta_1)$ then $D$ can be discarded without
changing the deduction closure as $C\md D\sg$. Note $\eta_2=
Res(L(\eta_1))$. By Theorem \ref{prime15}, $\pi(\gamma)=\pi(\eta_1)\se
\overline{L}(\eta_2)$. In general, $\pi(\gamma)=\pi(\eta_1)\se
\eta_i=Res(L(\eta_{i-1}))\se\overline{L}(\eta_{i-1})$. If the
algorithm terminates, at some stage then
$\overline{L}(\eta_i)=\overline{L}(\eta_{i-1})$ and
$Res(\overline{L}(\eta_i))= Res(\overline{L}(\eta_{i-1}))$. By Theorem
\ref{prime15}, $\pi(\gamma)= Res(\overline{L}(\eta_i))$. Since the
algorithm computes $\eta_{i+1}=Res(L(\eta_i))$,
$\pi(\gamma)=\eta_{i+1}$.  By Theorem \ref{prime16} we obtain the set
of prime implicates of $X_1\cup C$ as $\pi(X_1\cup C)=\pi(\pi(X_1)\cup
C)=\pi(\gamma)=\eta_{i+1}$. \hfill{$\Box$}
\end{proof}

\begin{example}\label{exp}
Let $X=\{\{Qy\},\{\neg Rf(x)b\},\{Px\vee Ryb\vee\neg
Qz\}\}$$=\eta_1(say)$ and $C$ be the clause $C=\{\neg Pa\vee\neg
Qz\}$$=\Sg$, another clause. Take $\Sg=\{C\}$
\end{example}

As computed in \cite{Manoj}, the set of prime implicates of the formula $X$
is $\pi(X)=\{\{Qy\},\{\neg Rf(x)b\}, \{Px\vee Rzb\},
\{Px\vee\neg Qz\}\}$, where the clause $Qy$ is associated with $\epsilon$,
$\neg Rf(x)b$ is associated with $\epsilon$, $Px\vee Rzb$ is associated with
$[y/z]$ and $Px\vee\neg Qz$ is associated with $[y/f(x)]$. \\ 

\noindent Let $\eta_1=\pi(X)\cup \Sg=\{\{Qy\},\{\neg Rf(x)b\},
\{Px\vee Rzb\}, \{Px\vee\neg Qz\},\{\neg Pa\vee\neg Qz\}\}$.\\

\noindent As the literal $Qy$ in $\{Qy\}$ and $\neg Qz$ in $\{\neg
Pa\vee\neg Qz\}$ are a pair of complementary literals with respect to
the substitution $[y/z]$,  the consensus between the clauses
$\{Qy\}$ and $\{\neg Pa\vee\neg Qz\}$ is $\{\neg Pa\}$. Here, the
substitution $[y/z]$ is a most general unifier.  Note that we can not
take consensus between $\{Px\vee Rzb\}$ and $\{\neg Pa\vee\neg Qz\}$,
between $\{Px\vee\neg Qz\}$ and $\{\neg Pa\vee\neg Qz\}$ as
composition of substitution are not well defined. Now updating $\Sg$ we
get, $\Sg= \{\{\neg Pa\vee\neg Qz\},\{\neg Pa\}\}$.  The new clause
formed is added to $\eta_1$ to get $L(\eta_1)$.\\

\noindent $L(\eta_1)=\{\{Qy\},\{\neg Rf(x)b\}, \{Px\vee Rzb\},
\{Px\vee\neg Qz\},\{\neg Pa\vee\neg Qz\}, \{\neg Pa\}\}$.\\

\noindent As $\{\neg Pa\}$ subsumes
$\{\neg Pa\vee\neg Qz\}$ in $L(\eta_1)$, we get the residue as\\

\noindent $\eta_2=Res(L(\eta_1))=\{\{Qy\},\{\neg Rf(x)b\}, \{Px\vee
Rzb\}, \{Px\vee\neg Qz\}, \{\neg Pa\}\}$.\\

\noindent Now the clause $\{Px\vee Rzb\}$ associated with $[y/z]$ and
$\{\neg Pa\}$ associated with $[y/z]$ contain a pair of complementary
literals.  The consensus between $\{Px\vee Rzb\}$ and $\{\neg Pa\}$
with respect to the substitution (mgu) $[y/z,x/a]$ is $Rzb$. Again we
can not take consensus between $\{Px\vee\neg Qz\}$ associated with
$[y/f(x)]$ and $\{\neg Pa\}$ associated with $[y/z]$ as composition of
substitution is not well defined. Now $\Sg=\{\{\neg Pa\}, \{Rzb\}\}$.
We now add the new clause to $\eta_2$ to get $L(\eta_2)$.\\

\noindent $L(\eta_2)=\{\{Qy\},\{\neg Rf(x)b\}, \{Px\vee Rzb\},
\{Px\vee\neg Qz\}, \{\neg Pa, Rzb\}\}$.\\

\noindent  As $\{Rzb\}$ subsumes $\{Px\vee Rzb\}$, the residue becomes\\

\noindent $\eta_3=Res(L(\eta_2))=\{\{Qy\},\{\neg Rf(x)b\},
\{Px\vee\neg Qz\}, \{\neg Pa\}, \{Rzb\}\}$.\\

\noindent We cannot take any more consensus among the clauses of
$\eta_3$ as the composition of substitution is not well defined
between clauses. So $\eta_3 = L(\eta_3)=\eta_4$.That is,\\

\noindent $\pi(X\cup C)=\pi(\pi(X)\cup C)= \{\{Qy\},\{\neg Rf(x)b\},
\{Px\vee \neg Qz\}, \{\neg Pa\}, \{Rzb\}\}$.












\section{Conclusion}

  In this paper, we have suggested an incremental algorithm to compute
  the set of prime implicates of a knowledge base $X$ and a clause
  $C$. We have also proved the correctness of the algorithm. In
  Example \ref{exp}, when new clauses or clause sets are added,
  computation of the prime implicates uses the primeness of the
  already computed clauses. The algorithm adds one clause at a time
  and compiles the enhanced knowledge base. A simple modification of
  the algorithm can be made to accomodate a set of clauses instead of
  one by putting {\it INCRPI} inside an iterative loop.

If we compute the prime implicates of $X\cup C$ directly by using the
algorithm from \cite{Manoj}, we obtain the same prime implicates,
though it involves wasteful computations. Efficiency of the proposed
algorithm {\it INCRPI} results in exploiting the properties of $\pi(X)$
instead of using $X$ directly.

\end{document}